\documentclass[aip,amsmath,amssymb,reprint,nofootinbib,cha]{revtex4-1}

\usepackage[utf8]{inputenc}
\usepackage{amsfonts}
\usepackage{amsthm}
\usepackage{graphicx}
\usepackage{dcolumn}
\usepackage{bm}
\usepackage{hyperref}
\usepackage{mathptmx}
\usepackage{gensymb}
\usepackage{xcolor}
\usepackage{array}

\newtheorem{theorem}{Theorem}
\newtheorem{remark}{Remark}

\graphicspath{{figures/}}

\usepackage[shortlabels]{enumitem}

\hypersetup{
        colorlinks=true,
        linkcolor=blue,
        citecolor=blue,
        filecolor=blue,
        urlcolor=blue,
}
\definecolor{rred}{rgb}{0.55, 0.0, 0.0}

\usepackage[normalem]{ulem}

\begin{document}

\title{Forecast error growth: A dynamic--stochastic
model}
\author{Eviatar Bach}
\email{eviatarbach@protonmail.com}
\affiliation{Department of Environmental Science and Engineering and Department of Computing and Mathematical Sciences, California Institute of Technology, Pasadena, California, 91125, USA}
\affiliation{Department of Meteorology and Department of Mathematics and Statistics, University of Reading, Reading, RG6 6ET, UK}
\affiliation{National Centre for Earth Observation, Reading, RG6 6ET, UK}

\author{Dan Crisan}
\email{d.crisan@imperial.ac.uk}
\affiliation{Department of Mathematics, Imperial College London, London, SW7 2AZ, UK}

\author{Michael Ghil}
\email{ghil@lmd.ipsl.fr}
\affiliation{Geosciences Department and Laboratoire de M\'et\'eorologie Dynamique (CNRS and IPSL), École Normale Sup\'erieure and PSL University, Paris, 75005, France}
\affiliation{Department of Atmospheric and Oceanic Sciences, University of
	California, Los Angeles, California, 90095, USA}
\affiliation{Department of Mathematics, Imperial College London, London, SW7 2AZ, UK}

\date{\today }

\begin{abstract}   
There is a history of simple forecast error growth models designed to capture the key properties of error growth in operational numerical weather prediction (NWP) models. We propose here such a scalar model that relies on the previous ones and incorporates multiplicative noise in a nonlinear stochastic differential equation (SDE). We analyze the properties of this SDE, including the shape of the error growth curve for small times and its stationary distribution, and prove well-posedness and positivity of solutions. Next, we fit this model to operational NWP error growth curves, and show good agreement with both the mean and probabilistic features of the error growth. These results suggest that the dynamic--stochastic error growth model proposed herein and similar ones could play a role in many other areas of the sciences that involve prediction.


\bigskip
\begin{center}
       {\it This paper is dedicated to the memory of Eugenia Kalnay and to her fundamental contributions to research in predictability and weather forecasting.}
\end{center}

\end{abstract}

\pacs{}
\maketitle
\preprint{AIP/123-QED}


\begin{quotation}
	Quantifying and understanding forecast error growth plays a key role in prediction in many fields of the natural and socio-economic sciences. Besides the extensive literature on the study of error growth in numerical weather prediction (NWP), this topic also plays an important role in understanding the limitations of prediction in other parts of the Earth system\cite{wirth_error_2000}, as well as in biology\cite{baird_increasing_2010,pei_counteracting_2017} and the social sciences\cite{farmer_how_2016}. Here we propose a simple stochastic model to describe forecast error growth. We analyze its mathematical properties and show that it exhibits a good fit to real NWP error growth curves.
\end{quotation}

\section{Introduction}\label{sec:intro}

\begin{figure}
    \centering
    \includegraphics[width=0.7\linewidth]{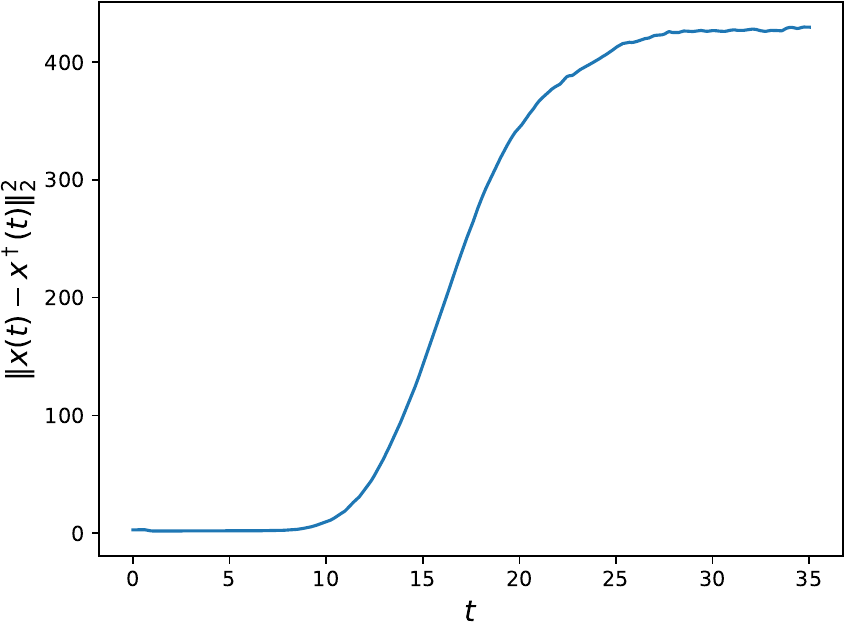}
    \caption{A typical pattern of error growth for a chaotic dynamical system. The system used here is the chaotic convection model of Lorenz (1963)\cite{lorenz_deterministic_1963}. The error growth is averaged over many initial conditions and a moving average filter is applied.}
    \label{fig:lorenz63_error}
\end{figure}
Consider a dynamical system
\begin{equation}\label{eq:ds}
    \mathrm{d}x = f(x)\,\mathrm{d}t + \Sigma^{1/2}(x)\,\mathrm{d}B,
\end{equation}
where $\Sigma$ is a symmetric positive definite covariance matrix and $B$ is a Wiener process. Given a true trajectory $x^\dagger(t)$ and a forecast trajectory $x(t)$, both generated by an equation like Eq.~\eqref{eq:ds}, one can measure the forecast error by
\begin{equation}
    e^2(t) = \|x(t) - x^\dagger(t)\|^2_2,
\end{equation}
where $\|\cdot\|_2$ is the Euclidean norm; other norms can also be considered. If $x^\dagger(0)\neq x(0)$ and the system is chaotic in the absence of noise, $e(t)$ will generally grow exponentially in time until reaching a saturation level; see Fig.~\ref{fig:lorenz63_error} for a typical example. If the noise realization differs between $x^\dagger(t)$ and $x(t)$, the noise will also contribute to the error growth. Moreover, a model typically differs from the system being modeled---i.e., $x^\dagger(t)$ and $x(t)$ come from similar but not identical systems---and this systematic model error contributes to the error growth as well.

\subsection{Stochastic treatment of error growth}\label{ssec:stochastic}

Suppose we have some uncertainty in the knowledge of the initial conditions $x^\dagger(0)$, represented by a probability distribution $\rho_0(x)$, such that $x^\dagger(0) \sim \rho_0(x)$. This probability distribution $\rho(x, t)$ will evolve under the dynamics according to the Fokker--Planck equation:
\begin{equation}
    \frac{\partial\rho}{\partial t} = -\nabla\cdot(\rho f) + \frac{1}{2} \nabla\cdot\bigl(\nabla\cdot(\rho \Sigma)\bigr),
\end{equation}
with $\rho(x, 0) = \rho_0(x)$. In the above equation, the divergence of a matrix $S$, such as $\rho\Sigma$, is defined, for any vector $a$, by the identity $(\nabla\cdot S)\cdot a \equiv \nabla\cdot(S^T a)$, and $\nabla\cdot b$ is the divergence of a vector $b$.

If we further assume that a forecast was initialized from $x(0)\sim \rho_0$, then $x(t)\sim\rho(x, t)$, and $e(t)$ will be a random variable whose properties can be derived from those of $\rho(t)$. For example, the expected squared error satisfies
\begin{subequations}
\begin{align}
    \mathbb{E}[e(t)^2] &= \mathbb{E}[\|x(t) - x^\dagger(t)\|_2^2]\\
    &= \mathbb{E}[\|x(t) - \mathbb{E}[x(t)]\|_2^2 + \|x^\dagger(t) - \mathbb{E}[x(t)]\|_2^2],\nonumber\\
    &= 2\text{tr}(\text{Cov}[x(t)]),
\end{align}
\end{subequations}
and the cross-covariance terms vanish because $x$ and $x^\dagger$ are independent and identically distributed.

If $\rho(t)$ converges to a unique invariant measure, that is, $\rho(t)\to\rho_\infty$ as $t\to\infty$, then $\text{Cov}[x(t)]\to C_\infty$, known in meteorological contexts as the climatological variance. When the forecast error reaches this climatological variance, the forecast has lost skill completely, and the forecast with the lowest error is the mean of $\rho_\infty$, with an expected squared error of $\text{tr}(C_\infty)$; see Leith (1974)\cite{leith_theoretical_1974}. Additional statistical properties of $e(t)$ can also be studied.

In ensemble prediction, such as that used nowadays in NWP\cite{kalnay_historical_2019}, a set of $M$ initial conditions drawn from $\rho_0$, $\{x^{(i)}\}_{i=1}^M$, are evolved in time under the dynamics Eq.~\ref{eq:ds}. For a perfect ensemble, each ensemble member will be statistically indistinguishable from the true trajectory \cite{palmer_ensemble_2006}. Thus, the error of each member will be a sample of the random variable $e(t)$ discussed above.

The preceding discussion shows that uncertainty in the initial conditions leads naturally to a probabilistic treatment of error growth, even in the case of a deterministic dynamical system where $\Sigma = 0$. Moreover, sources of error growth in deterministic systems that may nonetheless be modeled as stochastic include fluctuations in the finite-time Lyapunov exponents\cite{nicolis_probabilistic_1992}, the impact of small and fast scales on large and slow ones, and model error \cite{ehrendorfer_liouville_1994}. Other works that took a probabilistic view on forecast error growth include Balgovind et al. (1983)\cite{balgovind_stochastic-dynamic_1983}, Benzi and Carnevale (1989)\cite{benzi_possible_1989}, Ehrendorfer (1994)\cite{ehrendorfer_liouville_1994}, Ivanov et al. (1994)\cite{ivanov_prediction_1994}, and Chu and Ivanov (2002)\cite{chu_linear_2002}. Here we do not deal with these sources of error explicitly, but do consider their cumulative effect.

\subsection{Previous error growth models in NWP}
\label{ssec:class}

Considerable attention has been given to forecast error growth in the forecast--assimilation (FA) cycle of operational NWP models, and a number of scalar error growth models have been proposed so far. For illustration purposes, we start by considering here four different models for error growth in NWP, namely those of C. E. Leith \cite{leith_objective_1978}, E. N. Lorenz \cite%
{lorenz_atmospheric_1982}, A. Dalcher and E. Kalnay \cite%
{dalcher_error_1987}, and C. Nicolis\cite{nicolis_probabilistic_1992}; see also Appendix A of Crisan and Ghil (2023)\cite{crisan_asymptotic_2023} and Krishnamurthy (2019)\cite{krishnamurthy_predictability_2019} for reviews of the literature.

Letting $v(t) = e(t)^2$, the forecast error model of Leith (1978)\cite{leith_objective_1978} is given by the
scalar, linear inhomogeneous ordinary differential equation (ODE) 
\begin{equation}  \label{eq:Leith}
\frac{\text{d}v}{\text{d}t} = \alpha v + s,
\end{equation}
where $\alpha>0$ measures the rate of growth of small errors, $s$ is the systematic model error, and $t$ is the lead time. The solution to this ODE is given explicitly by 
\begin{equation}
v(t) = \left(v_0 + \frac{s}{\alpha}\right) e^{\alpha t} - \frac{s}{\alpha},
\end{equation}
where $v_0 = v(0)$ is the initial error. Note that short-time forecast errors grow exponentially and that the systematic model error acts to increase the coefficient of this growth. Leith's forecast error model can only apply for short-time error growth, since it does not saturate.

Lorenz (1982)\cite{lorenz_atmospheric_1982} proposed another model of forecast error growth that is governed by the following scalar ODE with quadratic right-hand side,
\begin{equation}  \label{eq:Lorenz}
\frac{\text{d}e}{\text{d}t} = ae(e_\infty - e),
\end{equation}
where $e_\infty$ is a saturation value. While Lorenz's model does include a nonlinear saturation term, it does not
incorporate systematic model error.

The error model of Dalcher and Kalnay (1987)\cite{dalcher_error_1987}, henceforth DK, combines the key features of the Leith and Lorenz models, 
\begin{equation}  \label{eq:DK}
\frac{\text{d}v}{\text{d}t} = (\alpha v + s)(1 - v/v_\infty);
\end{equation}
it thus includes both a saturation level $v_\infty$ and systematic model error $s$. For short-time error growth, we can take $v_\infty\to \infty$, recovering Leith's model. For $s = 0$, the model is similar to that of Lorenz. DK fit this error growth model to operational NWP forecasts from the European Centre for Medium-Range Weather Forecasting (ECMWF), finding close agreement out to 10 days. Note that $s$ has previously been taken to represent sources of error other than systematic model error as well \cite{zhang_what_2019}.

Nicolis (1992)\cite{nicolis_probabilistic_1992} proposed a stochastic error model of the form
\begin{equation}
	\text{d}\xi = (\sigma - g\xi)\xi\,\text{d}t + q\xi\,\text{d}W,\label{eq:nicolis}
\end{equation}
where $\xi$ is a measure of error along the unstable direction, $W$ is a Wiener process, and $g\geq 0$ and $q$ are constants. To the best of our knowledge, this work is the only stochastic model of error growth previously proposed. The formulation \eqref{eq:nicolis} is justified in the paper by considering a toy model of the atmosphere and introducing stochasticity coming from fluctuations in the finite-time Lyapunov exponents. This model was also analyzed by Chu et al. (2002)\cite{chu_probabilistic_2002}. If we remove the stochasticity by setting $q=0$, \eqref{eq:nicolis} has the same form as the Lorenz model. This model does not appear to have been calibrated to operational NWP error growth.

Other error growth models include those of Schubert and Suarez (1989)\cite{schubert_dynamical_1989}, of Stroe and Royer (1993)\cite{stroe_comparison_1993}, and of Reynolds et al. (1994)\cite{reynolds_random_1994}.

In the present paper, we build on the previous error growth models above by adding multiplicative noise, as in the Nicolis model, to the DK model.

The paper is structured as follows. In the next section, we introduce the proposed stochastic model with multiplicative noise and study its properties, including its stationary distribution and first passage times. We then fit the model to operational NWP error curves in Sec.~\ref{sec:NWP}. In Sec.~\ref{sec:conclude}, we discuss the results thus obtained. We also comment on the likelihood of these results applying to other areas of the sciences where prediction is important. In Appendix~\ref{sec:proof}, we prove well-posedness and positivity of our nonlinear stochastic model, with other model properties discussed in Appendices~\ref{sec:stationary} and \ref{sec:curvature}.

\section{A stochastic error growth model}\label{sec:growth}

\subsection{Stochastic model formulation}\label{ssec:SDE}

\begin{figure}[tbp]
	\includegraphics[scale=0.35]{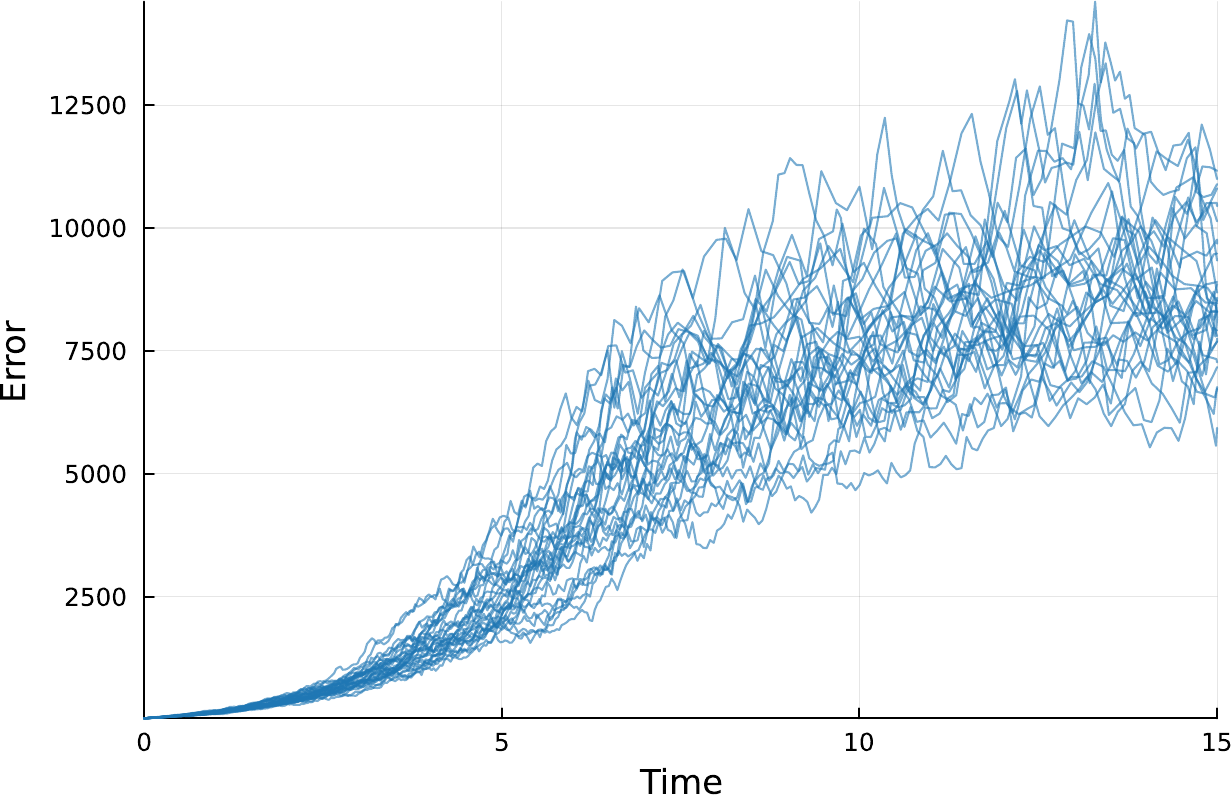}
	\caption{Sample realizations of the SDE~\eqref{eq:SDE_add}. Here $v(0) = 30$ and the parameters are picked to be $v_\infty = 9000$, $\alpha = 0.6$, $s = 80$, and $\sigma = 0.2$.
	}
    \label{fig:sample_paths}
\end{figure}
We introduce multiplicative noise into the DK \cite{dalcher_error_1987} model by writing the It\^o stochastic differential equation (SDE)
\begin{equation}
\mathrm{d}v=(\alpha v+s)(1-v/v_{\infty})\,\mathrm{d}t+\sigma v \,\mathrm{d}W,
\label{eq:SDE_add}
\end{equation}
with $W$ a standard Wiener process and $\sigma$ a constant. Figure~\ref{fig:sample_paths} shows some sample realizations from this SDE for a particular choice of parameter values.

This error model differs from the Nicolis model in Eq.~\eqref{eq:nicolis} as the latter has no equivalent to the term involving $s$. Note that $v_\infty$ no longer corresponds exactly to the saturation value, as it did in the deterministic case. In Sec.~\ref{ssec:fpe}, we consider two different definitions of a saturation value in the stochastic context: one is the maximizer of the stationary distribution, and the other is its expectation.

We chose multiplicative noise rather than additive noise since the former provided a better fit to operational NWP error curves; see Sec.~\ref{sec:NWP} here and the heuristic justifications for multiplicative noise in Nicolis (1992)\cite{nicolis_probabilistic_1992}. We will see in Appendix~\ref{sec:proof} that multiplicative noise leads, almost surely, to positive solutions given positive initial conditions. This property matters because the forecast error cannot be negative. A noise term that is a nonlinear function of $v$ could also be tried.

For $s = 0$ and as $v/v_\infty\to 0$, the drift term becomes linear in $v$ and the process becomes a geometric Brownian motion with solution
\begin{equation}
	v(t) = v(0)\exp\left( \left(\alpha - \frac{\sigma^2}{2} \right)t + \sigma W\right),
\end{equation}
with mean $\mathbb{E}[v(t)] = v(0)\exp(\alpha t)$ and variance $\operatorname{Var}[v(t)]= v(0)^2\exp(2\alpha t) \left( \exp(\sigma^2 t)-1\right)$. This formula approximately describes the small-time behavior of the process when $v(0)\ll v_\infty$, assuming that $s$ is small.

Our model \eqref{eq:SDE_add} also resembles that of stochastic logistic growth; see Liu and Wang (2013)\cite{liu_note_2013}.

\subsection{Nondimensionalization}

Setting 
\begin{subequations} \label{eq:ratios}
	\begin{align}
& T=\alpha ^{-1}, \quad V=s\alpha ^{-1}, \quad t^{\ast}=t/T, \label{nondim1} \\
& v^{\ast }=v/V, \quad v_{\infty }^{\ast } = v_{\infty }/V, \quad \beta = \sigma \alpha ^{-1/2}, \label{nondim2}
     \end{align}
\end{subequations}
we obtain the nondimensional equation: 
\begin{equation}
	\text{d}v^{\ast }=(v^{\ast }+1)\left( 1-\frac{v^{\ast }}{v_{\infty }^{\ast }}%
	\right) \text{d}t^{\ast }+\beta v^* \text{d}W^{\ast }  
	\label{eq:nondim}
\end{equation}
Thus, we have two nondimensional parameters, $\beta$ and $v_{\infty }^{\ast}$. The $\beta$ parameter here is directly related to the $\lambda$ parameter used by Chu et al. (2002)\cite{chu_probabilistic_2002} for the Nicolis model. Note that this nondimensionalization can only be applied when $s \neq 0$. When $s = 0$, we will work with the dimensional equation~\eqref{eq:SDE_add}.

In Appendix~\ref{sec:proof}, we show that if $v^*(0)>0$, Eq.~\eqref{eq:nondim} has, almost surely, a unique and positive global solution.

\subsection{Fokker--Planck equation and stationary distribution}\label{ssec:fpe}

The Fokker--Planck equation that corresponds to Eq.~\eqref{eq:nondim} is 
\begin{equation}
	\frac{\partial\rho}{\partial t^*} = -\frac{\partial}{\partial v^*}\left[(v^{\ast }+1)\left( 1-\frac{v^{\ast }}{v_{\infty }^{\ast }}%
	\right)\rho\right] + \frac{\partial^2}{\partial {v^*}^2}\left[\frac{\beta^2 {v^*}^2}{2}\rho\right].
\end{equation}
We would like to find a stationary distribution $\rho_\infty$ for which $\partial \rho_\infty/\partial t^* = 0$.

The detailed computations are provided in Appendix~\ref{sec:stationary}. The resulting stationary distribution is
\begin{equation}\label{eq:stationary_dist}
	\rho_\infty(v^*) = D\,\exp\!\left(-\frac{2 \left(\left(\beta ^2 v_{\infty}^*-v_{\infty}^*+1\right) \log (v^*)+\frac{\displaystyle v_{\infty}^*}{\displaystyle v^*}+v^*\right)}{\beta ^2 v_{\infty}^*}\right);
\end{equation}
here
\begin{equation}
	D =  \frac{{v_{\infty}^*}^{\displaystyle \frac{{v_{\infty}^*}^{-1}-1}{\beta ^2}+\frac{1}{2}}}{2 K_{m}\left( \displaystyle \frac{4}{{v_{\infty}^*}^{1/2} \beta ^2}\right)},\quad m = \frac{2{v_{\infty}^*}^{-1}-2}{\beta ^2}+1,
\end{equation}
and $K_n(z)$ is the modified Bessel function of the second kind, which satisfies $-y \left(n^2+z^2\right)+z^2 y''+z y'=0$.

\begin{remark}
Note that if $s = 0$, Eq.~\eqref{eq:stationary_dist} reduces to Eq.~(23) in Nicolis (1992)\cite{nicolis_probabilistic_1992}; i.e., in dimensional variables:
\begin{equation}
	\rho_\infty(v) = \frac{2^{\displaystyle\frac{2 \alpha }{\sigma ^2}-1} v^{\displaystyle\frac{2 \alpha }{\sigma ^2}-2} \left(\displaystyle\frac{\alpha }{\sigma ^2 v_\infty}\right)^{\displaystyle\frac{2 \alpha }{\sigma ^2}-1} e^{\displaystyle-\frac{2 \alpha  v}{\sigma ^2 v_\infty}}}{\Gamma \left(\frac{\displaystyle2 \alpha }{\displaystyle\sigma ^2}-1\right)},
\end{equation}
with $\alpha>\sigma^2/2$.
\end{remark}

We find, furthermore, that $\rho_\infty$ is unimodal, and maximized at
\begin{equation}
	v^*=\frac{1}{2} \left(-\beta ^2 v_{\infty}^*+[\left(\beta ^2 v_{\infty}^*-v_{\infty}^*+1\right)^2+4 v_{\infty}^*]^{1/2}+v_{\infty}^*-1\right).
\end{equation}
As $t\to\infty$, $\mathbb{E}[v^*(t)]$ saturates. When related to the dynamical error growth, this saturation value should correspond to twice the climatological variance; see section~\ref{ssec:stochastic}. 

The limit of the expectation is
\begin{equation}
    \lim_{t\to\infty}\mathbb{E}[v^*(t)] = (v_{\infty}^*)^{1/2}\frac{K_{m-1}\left(\frac{\displaystyle4}{\displaystyle{v_{\infty}^*}^{1/2} \beta ^2}\right)}{K_{m}\left(\frac{\displaystyle 4}{\displaystyle{v_{\infty}^*}^{1/2} \beta ^2}\right)},
\end{equation}
and the limit of the variance is
\begin{align}
    \lim_{t\to\infty}\text{Var}[v^*(t)] =& \frac{v_{\infty}^* K_{-m+2}\left(\frac{\displaystyle 4}{\displaystyle{v_{\infty}^*}^{1/2} \beta ^2}\right)}{K_{m}\left(\frac{\displaystyle 4}{\displaystyle{v_{\infty}^*}^{1/2} \beta ^2}\right)}\\
    &\qquad- \frac{v_{\infty}^* K_{m-1}\left(\frac{\displaystyle4}{\displaystyle{v_{\infty}^*}^{1/2} \beta ^2}\right){}^2}{K_{m}\left(\frac{\displaystyle 4}{\displaystyle {v_{\infty}^*}^{1/2} \beta ^2}\right)^2}\nonumber.
\end{align}

\subsection{First passage times and forecast skill horizons}\label{sec:first_passage}
\begin{figure}[tbp]
	\includegraphics[scale=0.4]{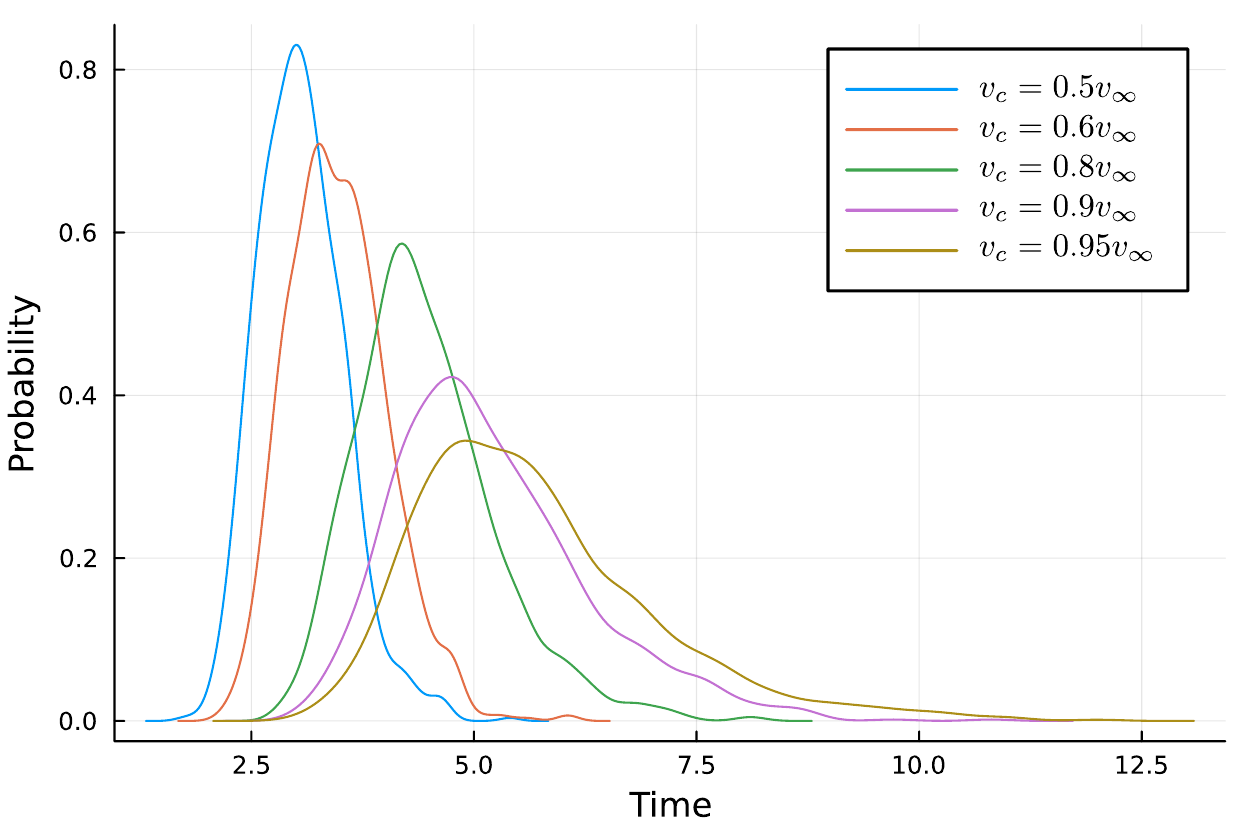}
	\caption{Empirical PDFs of first passage times for solutions of the SDE, for increasing threshold values, from $v_c = 0.5 v_{\infty}$ to $v_c = 0.95 v_{\infty}$. Time in nondimensional units. The PDFs are estimated from data using kernel density estimation \cite{hastie_elements_2009}; this is a likely reason for the small apparent deviations from unimodality.
	}
	\label{fig:exceedance_pdfs}
\end{figure}

A forecast is typically said to have lost skill when its error exceeds a certain threshold, e.g., some fraction of the climatological variance; see the discussion in section \ref{ssec:stochastic}. We can estimate the probability density function (PDF) of the times at which forecasts exceed a certain threshold $v_c$. These times have been previously referred to as first passage times\cite{ivanov_prediction_1994,chu_probabilistic_2002}. When the threshold is used to measure the loss of prediction skill, then these first passage times have also been referred to as forecast skill horizons\cite{buizza_forecast_2015}.

Figure~\ref%
{fig:exceedance_pdfs} shows empirical PDFs of first passage times for solutions of Eq.~\eqref{eq:SDE_add} with a given choice of parameters. We can see that these distributions tend to be positively skewed, i.e., the tail is towards larger first passage times. Moreover, the PDF flattens for higher values of the threshold.

Chu et al. (2002)\cite{chu_linear_2002}  examined in detail the distribution of first passage times for the Nicolis model of Eq.~\eqref{eq:nicolis}, including expressions for the moments of this distribution under various asymptotic regimes.

To summarize the skill horizon in a single number in this probabilistic perspective, we can consider:
\begin{enumerate}[(i), nosep]
\item The skill horizon of an average forecast, $\min \{t: \mathbb{E}[v(t)] > v_\text{c}\}$.
\item The average skill horizon of a forecast, $\mathbb{E}[\min \{t: v(t) > v_\text{c}\}]$.
\end{enumerate}
It may also be of interest to examine the tails of the first passage time PDFs, i.e., the probability of forecasts remaining skillful for an unusually long or short lead time.
  
\subsection{Initial curvature}

In Appendix~\ref{sec:curvature}, we examine the curvature of the expectation of the stochastic model at time 0; compare Fig.~\ref{fig:sample_paths} to Fig.~\ref{fig:sample_paths_concave} for an example of initially convex and initially concave expectation, respectively. We find that the curvature of the expectation, for a given deterministic initial condition $v_0$, is given by
\begin{align}
    \lim_{t\rightarrow 0}\frac{\mathrm{d}^{2}\mathbb{E}\left[ v\right] }{\mathrm{
d}t^{2}} &=\left( \alpha -\frac{s}{v_{\infty }}-2\frac{\alpha }{v_{\infty }}
v_{0}\right) (\alpha v_{0}+s)(1-v_{0}/v_{\infty })\label{eq:curvature}\\
&\qquad\qquad\qquad\qquad- \frac{\alpha}{v_\infty}\sigma^2 v_0^2.\nonumber
\end{align}
Thus, the curvature at time 0 may be either concave or convex, depending on the parameter values and the initial condition. When the curve is initially convex, it will still be asymptotically concave, implying that there is an inflection point.

Equation~\eqref{eq:curvature} also shows that noise acts to decrease the initial curvature. Furthermore, a large enough $\sigma$ guarantees that the error growth curve is concave throughout.

We note, however, that for the parameter values obtained from fits to NWP error growth in Sec.~\ref{sec:NWP}, the curvature is always positive and far from becoming negative even with nonzero $\sigma$.

\section{Error growth in the SDE model and in NWP models}
\label{sec:NWP}
To determine the usefulness of the stochastic model, we fit it to error growth curves from operational NWP.

\subsection{Forecasts and reanalysis data}

\begin{figure}
    \centering
    \includegraphics[scale=0.3]{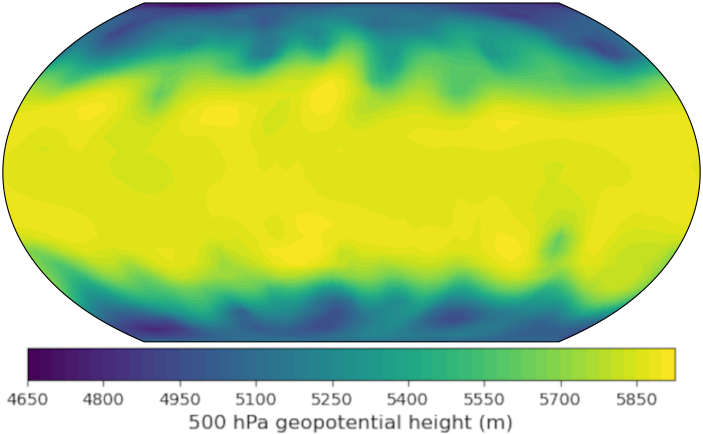}
    \caption{Example of a geopotential height forecast ensemble member.}
    \label{fig:gh}
\end{figure}

We used forecasts of 500~hPa geopotential heights from the TIGGE dataset of ensemble forecasts\cite{bougeault_thorpex_2010}. This is a common variable used in predictability studies\cite{benzi_possible_1989}; it corresponds roughly to the mid-level of the atmosphere, since the surface is, on average, at 1~000 hPa. Within this dataset, we use the global 50-member ensemble forecasts from the European Centre for Medium-Range Weather Forecasting (ECMWF), initialized daily at 00 UTC on the 1st through the 31st of January 2007. In these forecasts, the physics (parameterization) tendencies are stochastically perturbed to capture the uncertainty due to model error\cite{buizza_stochastic_1999}. See Fig.~\ref{fig:gh} for an example of the geopotential height field from the ensemble.

For each initialization time, we have ${50\choose 2} = 1225$ pairs of forecasts. For the identical twin experiments below, we compute the mean-squared difference between each pair as a function of lead time. Since the data is provided on a latitude--longitude grid, the mean is weighted by the area of each grid cell to account for variations in the distance corresponding to a degree of longitude. We aggregate over all the initialization dates, resulting in a total of 37~975 error curves.

We also compute the error of each ensemble member with respect to reanalysis, taken to be an estimate of the true atmospheric state. We used the ERA5 reanalysis dataset \cite{hersbach_era5_2020}. 

We compute the error of the forecasts, both for the identical twins and with respect to reanalysis, every 12 hours for 10 days of lead time. Each error curve thus consists of 21 points.

\subsection{Fitting procedure}\label{ssec:fitting}

We use ensemble Kalman inversion (EKI)\cite{iglesias_ensemble_2013} to fit the parameters of the SDE~\eqref{eq:SDE_add} to the NWP model error growth curves. EKI is derivative-free, which makes it an especially attractive method for parameter estimation in SDEs. It has previously been used for fitting SDE parameters in Schneider et al. (2021)\cite{schneider_learning_2021}. We only briefly describe EKI here; see also Chada (2022)\cite{chada_review_2022}, Calvello et al. (2024)\cite{calvello_ensemble_2024}, and Vernon et al. (2025)\cite{vernon_nesterov_2025} for more detailed introductions.

EKI is an iterative ensemble method for approximately minimizing with respect to $\theta$ a cost function of the form
\begin{equation}\label{eq:cost}
	(y - \mathcal{G}(\theta))^\top\Gamma^{-1}(y - \mathcal{G}(\theta)).
\end{equation}
Here, $\theta$ are the parameters one wishes to estimate, $y$ are the observed data with noise covariance $\Gamma$, and $\mathcal{G}$ is the forward model that is assumed to have generated the data.

As the observations $y$, we take the logarithm of the mean and standard deviation of the error growth curves over time; since the error growth curve is 21 elements long (10 days sampled every 12 hours), we obtain a vector with $n = 42$ components for each curve.

The logarithm of the data is taken in order to fit well both the initial and later portions of the error growth curve. We use the EnsembleKalmanProcesses.jl library\cite{dunbar_ensemblekalmanprocessesjl_2022}, with 100 ensemble members and 30 iterations, to calibrate the $v_\infty$, $\alpha$, $s$, and $\sigma$ parameters.

For all the calibrations, except for the case of wavenumbers $k \geq 10$ described below, the initial ensembles were drawn from the distributions
\begin{enumerate}[(a), nosep]
    \item $v_\infty$ with mean $1.4\times10^4$ and standard deviation $5~000$, constrained to be positive;
    \item $\alpha$ with mean 0.6 and standard deviation 0.3, constrained to be positive;
    \item $s$ with mean 200 with standard deviation 100, constrained to be positive; and
    \item $\sigma$ with mean 0.2 and standard deviation 0.1, constrained to be between 0 and 1.
\end{enumerate}
The means were chosen by rough inspection of solutions, while the standard deviations were chosen to be wide enough to explore a range of solutions. The constraints are implemented by calibrating in a transformed space and applying the inverse transformation to obtain the parameter estimates\cite{dunbar_ensemblekalmanprocessesjl_2022}: for the parameters constrained to be positive, we take the logarithm, and for the parameters constrained to be between 0 and 1 we apply the transformation $f(x) = \log(x / (1 - x))$. $\Gamma$ is set to be $0.25I_{n}$. For the $k\geq 10$ case, due to the smaller magnitudes of the errors, we set $\Gamma$ to be 40 times smaller, and the initial ensembles for $v_\infty$ and $s$ to have means and standard deviations 40 times smaller. The rest of the initial ensembles are left the same.

EKI begins with an initial ensemble of parameter sets, and generates another ensemble of parameter sets at each iteration. For each parameter set, we generate 300 realizations of the SDE using the SRIW1 integrator in the DifferentialEquations.jl package \cite{rackauckas_differentialequationsjl_2017} and compute the aforementioned statistics.

In order to capture the initial spread of forecasts, each realization is initialized with a value of initial mean-square difference from a pair of forecasts in the identical twin case, and with a mean-square difference from the ensemble mean in the reanalysis case. EKI proceeds to update the ensemble of parameter sets using $y$ and obtain the ensemble for the next iteration, converging towards parameter values that decrease the cost function \eqref{eq:cost}.

\subsection{Identical twins}

We first look at the error growth as measured by the divergence of identical twins, i.e., of two ensemble members initialized at the same time. This approach is used to measure the error growth due to chaotic evolution of the system, and it is often used in the literature \cite{dalcher_error_1987}. Note, however, that the error of individual ensemble members will be significantly larger than that of the ensemble mean, which is often used as the forecast. The impact of model error is partially captured in these forecasts by the stochastic perturbations to the physics tendencies.

\begin{figure}
	\includegraphics[scale=0.4]{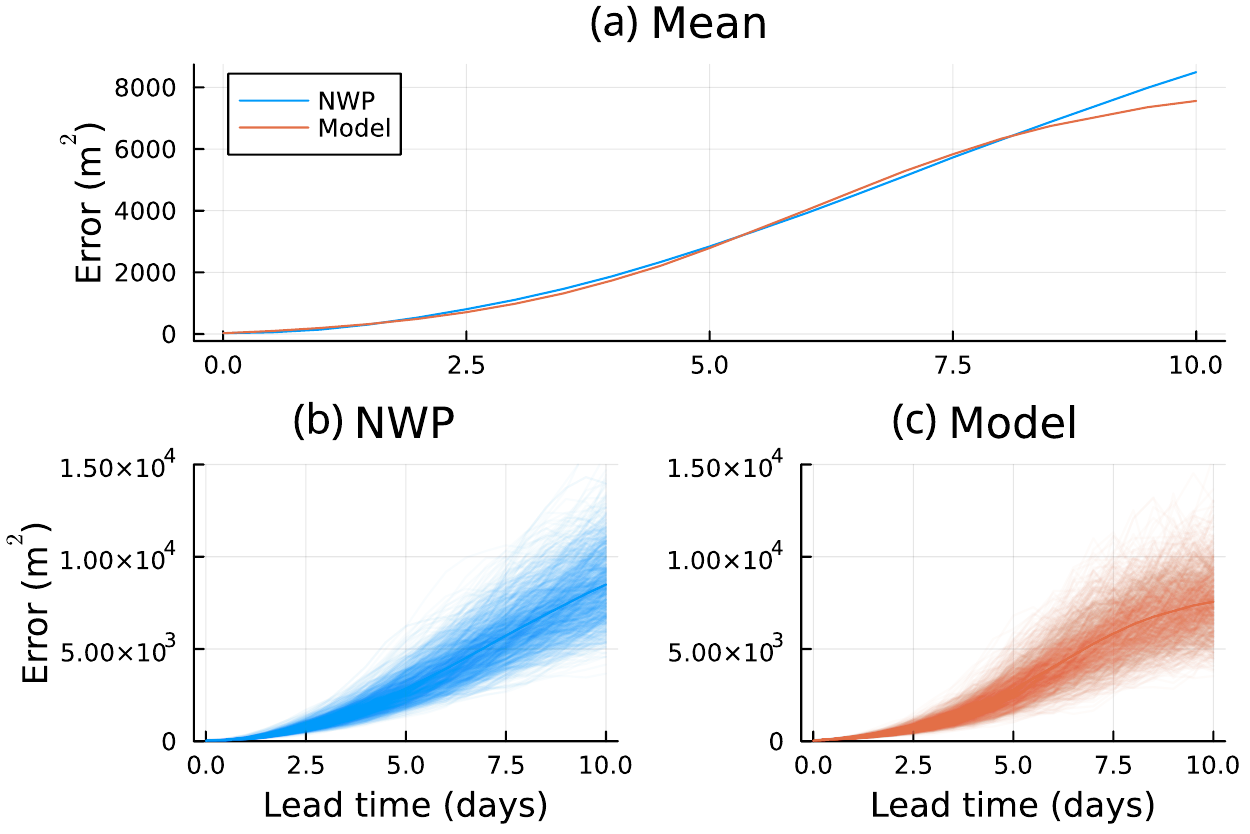}
	\caption{Mean error growth curves. (a) For identical twins in an operational NWP forecast (blue curve) and the SDE model (red curve); (b) error growth of the individual NWP forecast pairs; and (c) realizations of the fitted SDE model.}
	\label{fig:all_k}
\end{figure}

\begin{table}
\begin{tabular}[t]{m{9em}|llll}
	& $v_\infty$ (m$^2$) & $\alpha$ (d$^{-1}$) & $s$ (m$^2$ d$^{-1}$) & $\sigma$ (d$^{-1/2}$)\\
	\hline
	Identical twins & 8758 & 0.6062 & 109.7 & 0.2116 \\
	Identical twins $k \leq 9$ & 8305 & 0.6419 & 83.08 & 0.2236 \\
	Identical twins $k \geq 10$ & 197.6 & 0.7853 & 8.709 & 0.2050\\
    Identical twins, July & 8253 & 0.616 & 97.41 & 0.2120 \\
	Error with respect to reanalysis & 10530 & 0.4962 & 129.9 & 0.1859
\end{tabular}
\caption{Parameters of the SDE error growth model estimated using the EKI procedure described in Sec.~\ref{ssec:fitting}.}
\label{table:params}
\end{table}

Figure~\ref{fig:all_k} shows the fit of the error growth model to the curves obtained by NWP. The mean as well as the plume of individual trajectories exhibits a good fit, except for a slight underestimation at later times. The fitted SDE model appears to begin to saturate around day 9, while the NWP does not; this could be remedied by using lead times beyond 10 days in the fitting procedure. The parameters are given in the first row of Table~\ref{table:params}; we also give the parameters for the fit for forecasts from the same system but initialized from the 1st through the 31st of July 2007 in the fourth row.

In comparison to our results, Dalcher \& Kalnay (1987)\cite{dalcher_error_1987} found $v_\infty = 11076$ m$^2$, $\alpha = 0.43$ d$^{-1}$, and $s = 439$ m$^2$ d$^{-1}$ for boreal winter and $v_\infty = 9856$ m$^2$, $\alpha = 0.44$ d$^{-1}$, and $s = 597$ m$^2$ d$^{-1}$ for boreal summer. As expected, $v_\infty$ is smaller in our results, due to the improvements in forecasting. As in the Dalcher \& Kalnay (1987) results, $v_\infty$ was slightly smaller for summer than winter, and $\alpha$ was almost identical between the two seasons. However, we find a slightly smaller $s$ for summer than winter, while Dalcher \& Kalnay (1987) found the opposite.

For the identical twins, the best fits being provided by $s \neq 0$ can be explained by the stochastic physics, which are meant to capture the impact of systematic error on the growth in uncertainty. However, the fact that even for identical twins without stochastic physics the best fits were provided by $s\neq 0$ in Dalcher \& Kalnay (1987)\cite{dalcher_error_1987} and Zhang et al. (2019)\cite{zhang_what_2019} suggests that $s$ can lead to a more realistic shape of the error growth curve even in the absence of this error.

\begin{figure}[tbp]
	\includegraphics[scale=0.27]{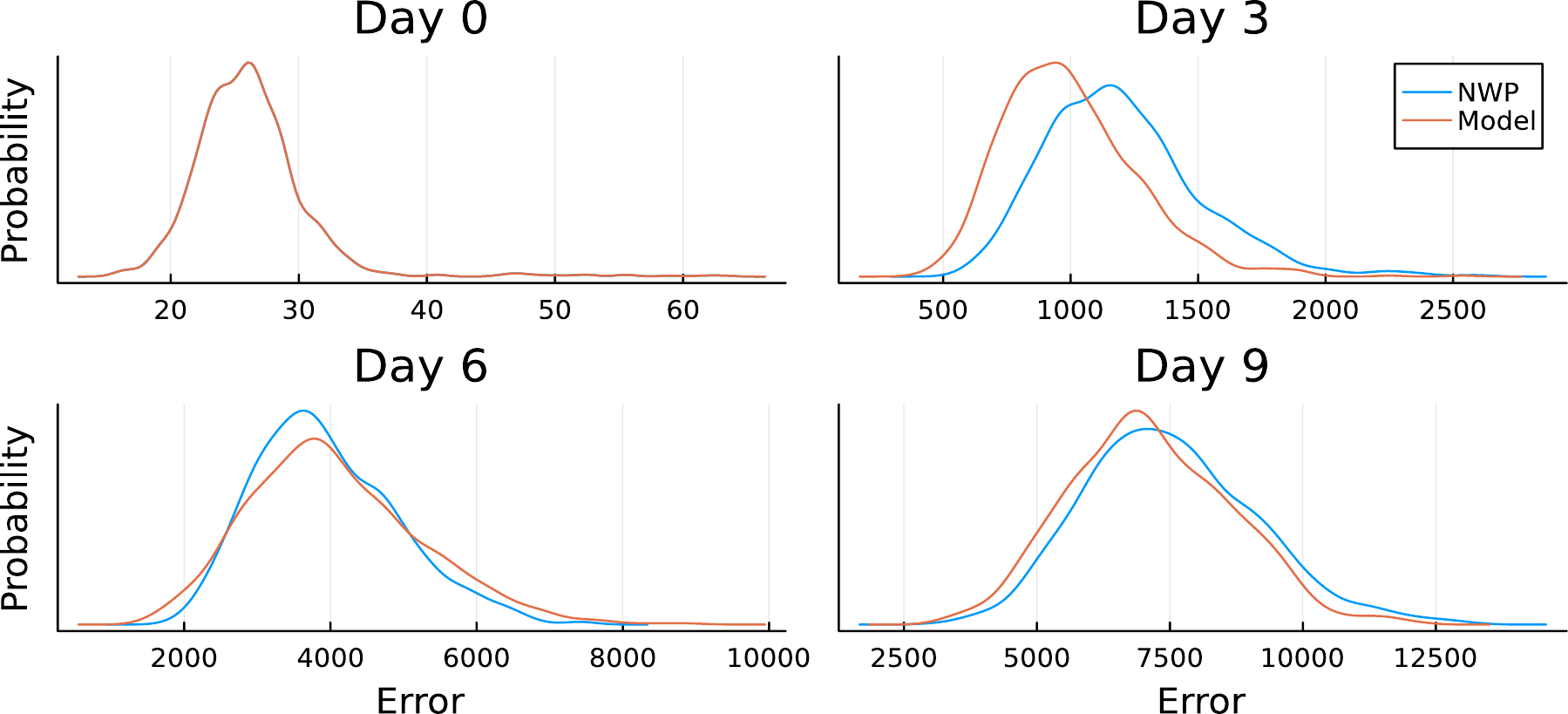}
	\caption{PDFs of the error at increasing lead times, for both the NWP data (blue curves) and the fitted SDE model (red curves). The PDFs are estimated from data using kernel density estimation.}
	\label{fig:pdfs}
\end{figure}

Figure~\ref{fig:pdfs} presents the time-evolving PDFs of trajectories, and the fitted model again matches closely the NWP error curves. At 3 days, however, the SDE model underestimates the error systematically; by 6 days this is no longer the case.

\begin{figure}
	\includegraphics[scale=0.4]{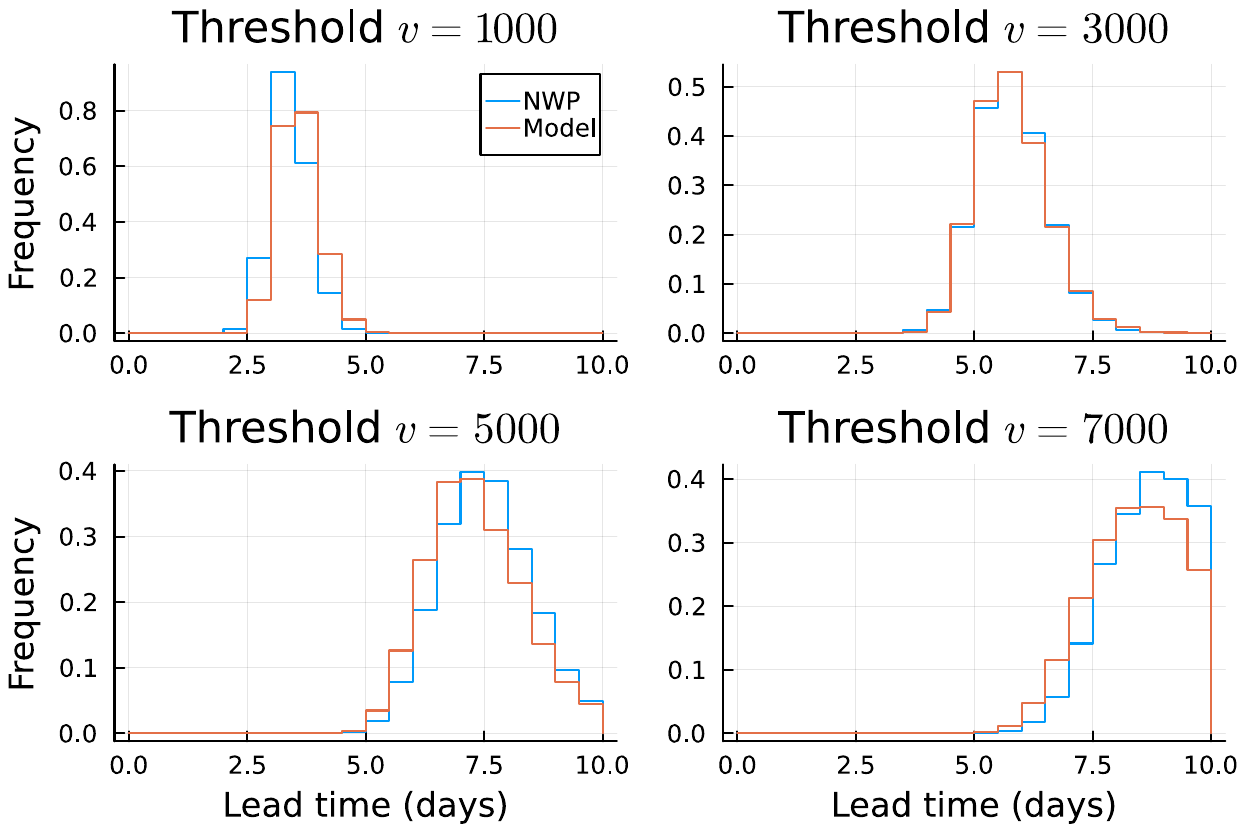}
	\caption{Histograms of first passage times of NWP error growth and of the fitted SDE model at different threshold values. Error growth curves with a first passage time greater than 10 days were excluded.}
 \label{fig:first_passage}
\end{figure}

Figure~\ref{fig:first_passage} displays histograms of first passage times computed for different thresholds for NWP as well as for the fitted SDE model. The histograms are closely matched. Both also flatten at higher thresholds and are positively skewed, as suggested in Sec.~\ref{sec:first_passage}.

\subsection{Scale analysis}

We repeated the identical twin experiments, but filtering the fields into low ($k\leq 9$) and high ($k\geq 10$) wavenumbers to measure the impact of spatial scale on error growth. For each of the wavenumber ranges, we filter the geopotential height fields and compute the mean-square error between the filtered fields.

The parameter values of the fits are summarized in Table~\ref{table:params}. The small scales exhibit a larger rate of error growth $\alpha$, saturate more quickly, and have a similar degree of stochasticity in the error growth as the large scales. The model generally appears to fit the large-scale error growth better than the small scales (not shown).

\subsection{Error with respect to reanalysis}

\begin{figure}
	\centering
	\includegraphics[scale=0.4]{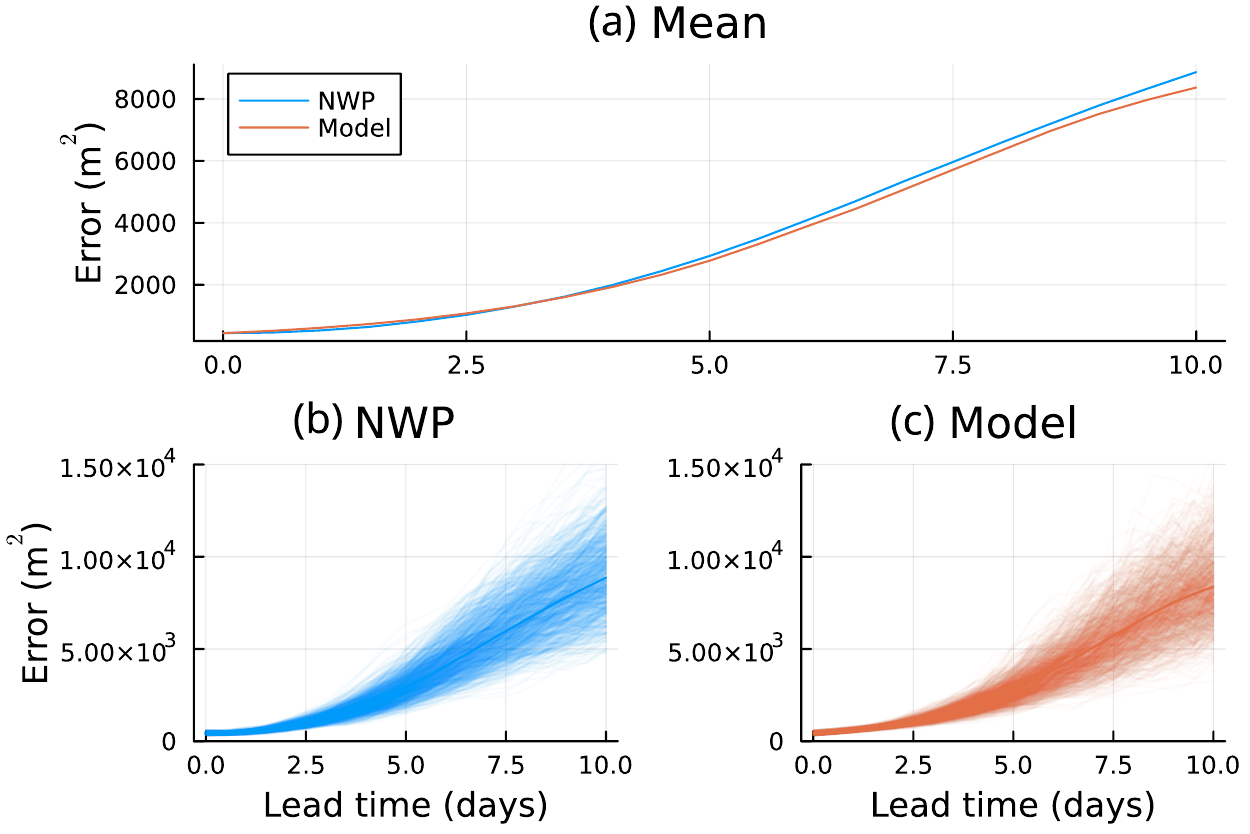}
	\caption{Same as Fig.~\ref{fig:all_k} but vs. reanalysis. (a) For NWP forecasts with respect to reanalysis (blue curve) and the SDE model (red curve); (b) error growth of the individual NWP forecasts; and (c) realizations of the fitted SDE.}
	\label{fig:systematic}
\end{figure}

We also compute the error curves of the NWP forecasts with respect to reanalysis, which includes the impact of systematic error. This is shown in Fig.~\ref{fig:systematic}, and the parameters given in Table~\ref{table:params}. Again, the SDE model exhibits a good fit to NWP error growth. Because of the impact of systematic error, $v_\infty$ is considerably larger than for the identical twin experiments. The best fit is provided by a smaller $\alpha$ but larger $s$ than for identical twins. The larger $s$ is expected since the systematic error is now directly contributing to the error growth.

\section{Concluding remarks}\label{sec:conclude}

We introduced a novel stochastic model of forecast error growth \eqref{eq:SDE_add}, building upon previous deterministic and stochastic models. Our nonlinear SDE model fits well the forecast error growth of the ECMWF's NWP model, accurately capturing both its mean and probabilistic aspects. Hence, the SDE model can be applied to quantify predictability in a probabilistic way.

Left for future work is the detailed analysis of the physical mechanisms leading to stochastic error growth, the application of this model to the forecast error growth of other systems, and its application to the analysis of predictability in the tropics compared to the extratropics\cite{bach_local_2019}. While here we looked at the error growth of individual ensemble members in NWP, it may also be interesting to model the growth of probabilistic forecast skill scores, such as the continuous ranked probability score (CRPS)\cite{hersbach_decomposition_2000}.

\section*{Data availability statement}

Data sharing is not applicable to this article as no new data were created or analyzed in this study.

\section*{Acknowledgments}

E. Bach was supported by the Foster and Coco Stanback Postdoctoral Fellowship at the California Institute of Technology. The work of D. Crisan has been partially supported by European Research Council (ERC) Synergy grant STUOD-DLV-856408. This paper is a ClimTip contribution; the Quantifying climate tipping points and their impacts (ClimTip) project has received funding from the European Union’s Horizon research and innovation programme under grant agreement No. 101137601. M. Ghil also received support from the French Agence Nationale de la Recherche (ANR) project TeMPlex under grant award ANR-23-CE56-1214 0002.

\bibliography{error_growth_sde}

\appendix

\section{The stochastic model: positivity and well posedness}
\label{sec:proof}

\begin{theorem}
	Equation~\eqref{eq:nondim} has a unique global solution $\{v^{\ast }(t^*): t^*\ge 0\}$, $P$-almost surely for any initial value $v^{\ast } (0) >0$. Moreover the solution remains positive  $P$-almost surely.  
\end{theorem}

\begin{proof}
First note that the coefficients of Eq.~\eqref{eq:nondim} are locally Lipschitz. It follows that, for any initial value $v^{\ast } (0) \in \mathbb{R}$, Eq.~\eqref{eq:nondim} has a unique maximal local strong	solution on $\left[0, \tau\right)$, where $\tau$ is the explosion time; see, for example Theorem~3.2.2, p. 95 in Mao (2007)\cite{mao_stochastic_2007}. 

We show next that $v^{\ast }$ cannot explode in finite time. More precisely, we show that $P$-almost surely $\tau =\infty $; in other words, Eq.~\eqref{eq:nondim} has a global solution, $P$-almost surely. To do this, we prove that $v^{\ast }$ is bounded from above and from below by stochastic processes that are finite on the entire half line. Moreover, the lower bound process has sample paths of any solution that will never reach the origin, $P$-almost surely.    

\textbf{Lower bound}. Observe that 
	\begin{equation*}
		(v^{\ast }+1)\left( 1-\frac{v^{\ast }}{v_{\infty }^{\ast }}\right) +\chi (v^{\ast })^{2}=\left(
		v^{\ast }\left( \frac{\frac{\displaystyle 1}{\displaystyle v_{\infty }^{\ast }}-1}{2}\right) -1\right) ^{2},
	\end{equation*}
where 
	\begin{equation*}
		\chi =\frac{\left( \frac{\displaystyle 1}{\displaystyle v_{\infty }^{\ast }}-1\right) ^{2}}{4}+\frac{1%
		}{v_{\infty }^{\ast }},
	\end{equation*}
hence
	\begin{equation*}
		(v^{\ast }+1)\left( 1-\frac{v^{\ast }}{v_{\infty }^{\ast }}\right) \geq -\chi (v^{\ast })^{2}.
	\end{equation*}
By applying  Proposition 2.18 in Karatzas and Shreve (1998)\cite{karatzas_brownian_1998} we get that, $P$-almost surely, $v^{\ast }\geq v^*_\text{low}$, where $v^*_\text{low}$ is the solution of the equation
\begin{equation}
		\text{d}v^*_\text{low}=-\chi \left( v^*_\text{low}\right) ^{2}\text{d}t^*+\beta v^*_\text{low}\,
		\text{d}W^*.  \label{multlower}
\end{equation}
We emphasize that the inequality $v^*\geq v^*_\text{low}$ holds pathwise and not just in probability.

Equation~\eqref{multlower} has a unique maximal solution (again, say by Theorem 3.2.2, p. 95 in Mao (2007)\cite{mao_stochastic_2007}), that satisfies the identity 
\begin{equation*}
		v^*_\text{low}(t^*)=v^\ast(0)\exp \left( \beta W_{t^*}-\frac{1}{2}\beta
		^{2}t^*-\int_{0}^{t^*}\chi v^*_\text{low}(s)\text{d}s\right) > 0.
	\end{equation*}
In particular, $v^*_\text{low}$ remains positive at all times and therefore
	\begin{align*}
		v^*_\text{low}(t^*)&=v^\ast(0)\exp \left( \beta W_{t^*}-\frac{1}{2}\beta
		^{2}t^*-\int_{0}^{t^*}\chi v^*_\text{low}(s)\text{d}s\right)\\ &<v^\ast(0)\exp \left( \beta W_{t^*}\right)\\
		&<\infty.
							\end{align*}
It follows that Eq.~\eqref{multlower} has a unique \emph{global} solution. Hence $v^{\ast}$ likewise stays positive on $[0,\tau)$ and, therefore, the explosion (if any) can only happen at $+\infty $. We show next that this is not possible either.

\textbf{Upper bound}. Observe that
\begin{equation*}
		(v^{\ast}+1)\left( 1-\frac{v^{\ast}}{v_{\infty }^{\ast }}\right) +\gamma v^{\ast}-1=-\frac{1}{%
			v_{\infty }^{\ast }}(v^{\ast})^{2},
	\end{equation*}
where
	\begin{equation*}
		\gamma =\frac{1}{v_{\infty }^{\ast }}-1.
	\end{equation*}
Hence
	\begin{equation*}
		(v^{\ast}+1)\left( 1-\frac{v^{\ast}}{v_{\infty }^{\ast }}\right) \leq -\gamma v^{\ast}+1.
	\end{equation*}
Again, by applying Proposition 2.18 in Karatzas and Shreve (1998)\cite{karatzas_brownian_1998} we get that, $P$-almost surely, $v^{\ast}\leq v^*_\text{high}$, where $v^*_\text{high}$ is the solution of the equation
\begin{equation}
		\text{d}v^*_\text{high}=\left( -\gamma v^*_\text{high}+1\right) \text{d}t^*+\beta v^*_\text{high}%
		\text{d}W^*. \label{multhigh}
	\end{equation}
For linear or affine coefficients, Equation~\eqref{multhigh} has the unique global solution
\begin{equation*}
		v^*_\text{high}(t^*)=v^\ast(0)\exp \left( \beta W_{t^*}-\frac{1}{2}\beta
		^{2}t^*-\gamma t^*+t^*\right)  <\infty.
	\end{equation*}
In particular it does not explode in finite time with probability 1, and therefore $v^\ast$, the solution of Eq.~\eqref{eq:nondim}, will not explode at $+\infty $. The global existence of $v^\ast$, as well as its positivity, are now ensured. \end{proof}

\section{The stochastic model's stationary distribution}\label{sec:stationary}

The Fokker--Planck equation that corresponds to Eq.~\eqref{eq:nondim} is 
\begin{equation*}
	\frac{\partial\rho}{\partial t^*} = -\frac{\partial}{\partial v^*}\left[(v^{\ast }+1)\left( 1-\frac{v^{\ast }}{v_{\infty }^{\ast }}%
	\right)\rho\right] + \frac{\partial^2}{\partial {v^*}^2}\left[\frac{\beta^2 {v^*}^2}{2}\rho\right].
\end{equation*}
We would like to find a stationary distribution $\rho_\infty$, namely solve $\partial \rho_\infty/\partial t^* = 0$.

We first take the integral with respect to $v^*$:
\begin{equation}
	-(v^{\ast }+1)\left( 1-\frac{v^{\ast }}{v_{\infty }^{\ast }}\right)\rho + \frac{\partial}{\partial {v^*}}\left[\frac{\beta^2 {v^*}^2}{2}\rho\right] = C,\label{eq:fp_integrated}
\end{equation}
where $C$ is an integration constant.
The solution to Eq.~\eqref{eq:fp_integrated} is then
\begin{equation*}
	\rho_\infty(v^*) = A(v^*) \exp\!\left(-\frac{2 \left(\left(\beta ^2 v_{\infty}^*-v_{\infty}^*+1\right) \log (v^*)+\frac{\displaystyle v_{\infty}^*}{\displaystyle v^*}+v^*\right)}{\beta ^2 v_{\infty}^*}\right),
\end{equation*}
where
\begin{align}\label{eq:A}
	A(v^*) =& \int _1^{v^*}\frac{2 C \exp \left(-\frac{\displaystyle 2\left((v_{\infty}^*-1) \log (V)-V-\frac{v_{\infty}^*}{V}\right)}{\displaystyle\beta ^2 v_{\infty}^*}\right)}{\beta ^2}\text{d}V\\
	&\qquad+D,\nonumber
\end{align}
and $D$ is another integration constant.

Since the limit of the integrand in Eq.~\eqref{eq:A} is infinite as $v^*\to\infty$ unless $C=0$, a necessary condition for $\rho_\infty(v^*)\to 0$ as $v^*\to\infty$ is that $C = 0$. Then,
\begin{equation*}
	\rho_\infty(v^*) = D\,\exp\!\left(-\frac{2 \left(\left(\beta ^2 v_{\infty}^*-v_{\infty}^*+1\right) \log (v^*)+\frac{\displaystyle v_{\infty}^*}{\displaystyle v^*}+v^*\right)}{\beta ^2 v_{\infty}^*}\right).
\end{equation*}
The constant $D$ can be eliminated using the normalization condition $\int_0^\infty \rho_\infty(v^*)\,dv^* = 1$:
\begin{equation*}
	D = \frac{{v_{\infty}^*}^{\displaystyle\frac{{v_{\infty}^*}^{-1}-1}{\beta ^2}+\frac{1}{2}}}{2 K_{\displaystyle\frac{2{v_{\infty}^*}^{-1}-2}{\beta ^2}+1}\left(\frac{\displaystyle 4}{\displaystyle{v_{\infty}^*}^{1/2} \beta ^2}\right)},
\end{equation*}
where $K_n(z)$ is the modified Bessel function of the second kind, satisfying $-y \left(n^2+z^2\right)+z^2 y''+z y'=0$.

\section{Initial curvature}\label{sec:curvature}

\begin{figure}
	\includegraphics[scale=0.35]{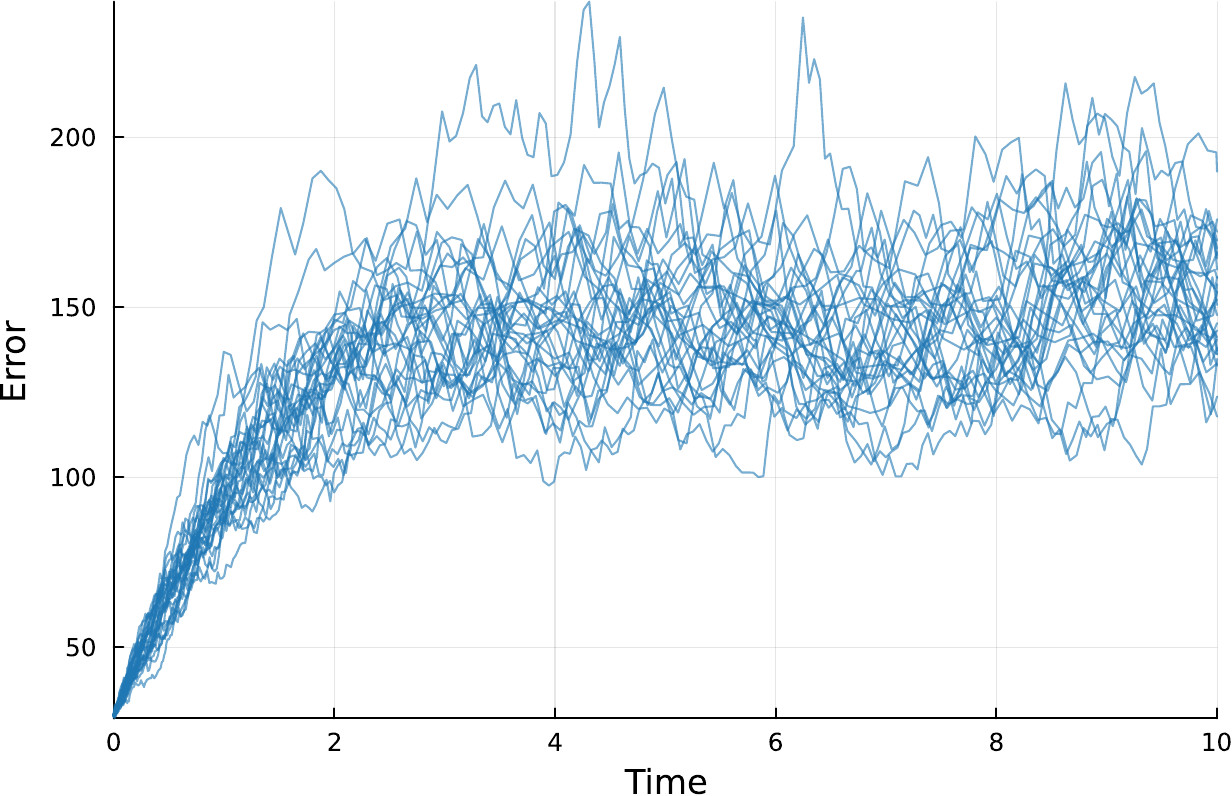}
	\caption{Sample realizations of the SDE~\eqref{eq:SDE_add} with concave expectation at time 0. Here $v(0) = 30$ and the parameters are $v_\infty = 150$, $\alpha = 0.6$, $s = 80$, and $\sigma = 0.2$.
	}
    \label{fig:sample_paths_concave}
\end{figure}
Figure~\ref{fig:sample_paths_concave} shows sample realizations of the SDE where the expectation is concave at time 0; compare to Fig.~\ref{fig:sample_paths} where the expectation is convex at time 0. We now determine the parameter values that lead the initial curvature to be either positive or negative.

Let us look at the deterministic DK model of Eq.~\eqref{eq:DK}
\begin{equation*}
\frac{\mathrm{d}v^{d}}{\mathrm{d}t}=(\alpha v^{d}+s)(1-v^{d}/v_{\infty }),
\end{equation*}
where we assume that the initial condition $v_{0}\leq v_{\infty}.$ Then
\begin{equation*}
\frac{\mathrm{d}^{2}v^{d}}{\mathrm{d}t^{2}}=\left( \alpha -\frac{s}{%
v_{\infty }}-2\frac{\alpha }{v_{\infty }}v^{d}\right) \frac{\mathrm{d}v^{d}}{%
\mathrm{d}t}.
\end{equation*}
Note that asymptotically 
\begin{equation*}
\lim_{v^{d}\rightarrow v_{\infty }}\left( \alpha -\frac{s}{v_{\infty }}-2%
\frac{\alpha }{v_{\infty }}v^{d}\right) =-\alpha-\frac{s}{v_{\infty }}<0.
\end{equation*}
In other words, the function is concave in the asymptotic regime. We have two
cases,
\begin{itemize}
\item If $\alpha -\frac{s}{v_{\infty }}-2\frac{\alpha }{v_{\infty }}%
v_{0}\leq 0$, then $\frac{\mathrm{d}^{2}v^{d}}{\mathrm{d}t^{2}}<0$ on the interval $(0,\infty )$ and therefore the function is strictly concave on the interval $(0,\infty )$.
\item If $\alpha -\frac{s}{v_{\infty }}-2\frac{\alpha }{v_{\infty }}v_{0}>0$, then there exists $t_{c}$ such that the function is convex on $\left(
0,t_{c}\right) $ and concave on $\left( t_{c},\infty \right)$; in other words, $t_{c}$ is an inflection point.
\end{itemize}

Returning to the stochastic model of Eq.~\eqref{eq:SDE_add}, we find that
\begin{equation*}
\frac{\mathrm{d}\mathbb{E}\left[ v\right] }{\mathrm{d}t}=(\alpha \mathbb{E}%
\left[ v\right] +s)(1-\mathbb{E}\left[ v\right] /v_{\infty })-\frac{\alpha }{%
v_{\infty }}\mathrm{Var}[v],
\end{equation*}
and 
\begin{equation*}
\frac{\mathrm{d}^{2}\mathbb{E}\left[ v\right] }{\mathrm{d}t^{2}}=\left(
\alpha -\frac{s}{v_{\infty }}-2\frac{\alpha }{v_{\infty }}\mathbb{E}\left[ v%
\right] \right) \frac{\mathrm{d}\mathbb{E}\left[ v\right] }{\mathrm{d}t}-%
\frac{\alpha }{v_{\infty }}\frac{\mathrm{dVar}[ v] }{\mathrm{d}t}.
\end{equation*}
We would like to get the limit at 0 of the derivative of the variance. Using It\^o's lemma,
\begin{eqnarray*}
\mathrm{d}v^{2} &=&2v(\alpha v+s)(1-v/v_{\infty })\mathrm{d}t+2\sigma 
v^2\mathrm{d}W+\sigma^{2}v^2\mathrm{d}t, \\
\frac{\mathrm{d}\mathbb{E}\left[ v^{2}\right] }{\mathrm{d}t} &=&2\mathbb{E}%
\left[ v(\alpha v+s)(1-v/v_{\infty })\right] +\sigma ^{2} \mathbb{E}[v^2],\\
\frac{\mathrm{d}\mathbb{E}\left[ v\right] ^{2}}{\mathrm{d}t} &=&2\mathbb{E}%
\left[ v\right] (\alpha \mathbb{E}\left[ v\right] +s)(1-\mathbb{E}\left[ v%
\right] /v_{\infty })\\
&&\qquad\qquad\qquad\qquad-2\mathbb{E}\left[ v\right] \frac{\alpha }{v_{\infty }}%
\mathrm{Var}[v] .
\end{eqnarray*}
If we assume a deterministic initial condition $v_0$ then 
\begin{eqnarray*}
\lim_{t\rightarrow 0}\frac{\mathrm{d}\mathbb{E}\left[ v^{2}\right] }{\mathrm{%
d}t} &=&2v_{0}(\alpha v_{0}+s)(1-v_{0}/v_{\infty }) + \sigma^2 v_0^2, \\
\lim_{t\rightarrow 0}\frac{\mathrm{d}\mathbb{E}\left[ v\right] ^{2}}{\mathrm{%
d}t} &=&2v_{0}(\alpha v_{0}+s)(1-v_{0}/v_{\infty }), \\
\lim_{t\rightarrow 0}\frac{\mathrm{dVar}\left[ v\right] }{\mathrm{d}t}
&=& \sigma^2 v_0^2, \\
\lim_{t\rightarrow 0}\frac{\mathrm{d}^{2}\mathbb{E}\left[ v\right] }{\mathrm{%
d}t^{2}} &=&\left( \alpha -\frac{s}{v_{\infty }}-2\frac{\alpha }{v_{\infty }}%
v_{0}\right) (\alpha v_{0}+s)(1-v_{0}/v_{\infty })\\
&& - \frac{\alpha}{v_\infty}\sigma^2 v_0^2 \quad
< \quad \lim_{t\rightarrow 0}\frac{\mathrm{d}^{2}v^{d}}{\mathrm{d}t^{2}}.
\end{eqnarray*}
Hence, $\lim_{t\rightarrow 0} {\mathrm{d}^{2}\mathbb{E}\left[ v\right] }/{%
\mathrm{d}t^{2}}$ will be negative for sufficiently high $\sigma ^{2}$; in other words, the expectation is a concave function in a neighborhood of 0.

\end{document}